\documentclass[a4paper,reqno]{amsart}

\usepackage{amsmath}

\usepackage{hyperref}

\newcommand{\R}{\mathbb{R}}
\renewcommand{\vec}[1]{\mathbf{#1}}

\newcommand{\norm}[1]{\|#1\|}
\newcommand{\secref}[1]{\textsc{Section \ref{#1}}}
\newcommand{\defref}[1]{\textbf{Definition \ref{#1}}}
\newcommand{\propref}[1]{\textbf{Proposition \ref{#1}}}
\newcommand{\corref}[1]{\textbf{Corollary \ref{#1}}}

\newtheorem{proposition}{Proposition}
\newtheorem{corollary}{Corollary}
\newtheorem{definition}{Definition}

\DeclareMathOperator{\Tr}{Tr}
\DeclareMathOperator{\im}{i}

\begin{document}
\title[quantum states as points in a probability simplex]{Representation of quantum states as points in a probability simplex associated to a SIC-POVM}
\author{Jos\'e Ignacio Rosado}

\email{joseirs@hotmail.com}


\date{\today}

\begin{abstract}
    The quantum state of a $d$-dimensional system can be represented by a probability distribution over the $d^2$ outcomes of a Symmetric Informationally Complete Positive Operator Valued Measure (SIC-POVM), and then this probability distribution can be represented by a vector of $\R^{d^2-1}$ in a $(d^2-1)$-dimensional simplex, we will call this set of vectors $\mathcal{Q}$. Other way of represent a $d$-dimensional system is by the corresponding Bloch vector also in $\R^{d^2-1}$, we will call this set of vectors $\mathcal{B}$. In this paper it is proved that with the adequate scaling $\mathcal{B}=\mathcal{Q}$. Also we indicate some features of the shape of $\mathcal{Q}$.
\end{abstract}

\maketitle

\section{Introduction}
In quantum mechanics a quantum state is described by a density operator, $\rho$, but there are alternative descriptions. One alternative description is to parameterize the space of density matrices with Bloch vectors and study the structure formed by these vectors, $\mathcal{B}$, this is done thoroughly in \cite{BVN} and \cite{BVSKK}. Other possible description is provided by the probabilities $p_i=\Tr(E_i\rho)$, where $E_i$ are the elements of a SIC-POVM \cite{SICQM}, an informationally complete and symmetric POVM, with the minimal number of elements. This is the description of quantum states chosen in Quantum-Bayesianism \cite{QBRQS}.

According to the quantum-Bayesian approach to quantum foundations, see for example \cite{QMQI}, probabilities $p_i$ represent an agents Bayesian degrees of belief. When we represent the probabilities of a SIC-POVM as points in the corresponding probability simplex, $\Delta$, we will see that these probabilities are not arbitrary, not any point of $\Delta$ can represent a quantum state, only a proper subset $\mathcal{Q}\subset\Delta$. The problem is then to understand the structure of $\mathcal{Q}$, if it is possible, in physical terms. Why are the beliefs of our agent constrained in this way? I think a hint to the answer to this question is to realize that the structure of $\mathcal{Q}$ is agent independent, of course for a particular experiment two agents can differ in the assignments of probabilities for the outcomes of a SIC-POVM, but the two distributions will be represented by points in $\mathcal{Q}$. Then the structure of $\mathcal{Q}$ is saying us something about the external world, at least about the intersubjective world.

The main result of this paper is demonstrate that $\mathcal{B}=\mathcal{Q}$, with suitable scalings. We therefore can translate the results obtained in the study of $\mathcal{B}$ to probability concepts, and I think this is a useful translation because in terms of probabilities we can made use of the tools of Quantum Information Theory and our physical intuition in trying to understand the \emph{why} of the $\mathcal{Q}$ structure.

The structure of this paper is as follows. In \secref{s:simplexes} we review some basic facts about simplexes as geometrical objects that can represent probability distributions. In \secref{s:BlochQS} we review the Bloch representation of quantum states and the conditions these vectors satisfy. In \secref{s:SICBloch} we give the Bloch representation of a SIC-POVM, and see what conditions satisfy the corresponding Bloch vectors. In \secref{s:PQS} we construct $\mathcal{Q}$ and demonstrate its equality with $\mathcal{B}$, we also give some features of the shape of this set.

\section{Simplexes and probability distributions}\label{s:simplexes}
This section contains known facts about simplexes and probability distributions, but serves to collect useful results and to fix some conventions in this paper.
\begin{definition}\label{d:simplex}
    An $n$-dimensional and regular simplex in $\R^n$ is any set, $\Delta_{n+1}$, that can be defined as follows
    \begin{align*}
        \Delta_{n+1}&=\left\{\vec s=\sum_{i=1}^{n+1}p_i\vec t_i\;:\; p_i\in [0,1] ,\;\sum_{i=1}^{n+1}p_i=1,\;\; \vec t_i\in\R^n\;\;\text{and}\right.\\
        \\
        &\quad\left. \vec t_i\cdot\vec t_j=(a-b)\delta_{ij}+b\,,\;\;i,j\in\{1,2,\ldots,n+1\}\;\;\text{and}\;\; b\neq a\neq 0\right\}.
    \end{align*}
    That is, $\Delta_{n+1}$ is the convex hull of the set of vectors $V=\{\vec t_i\}_{i=1}^{n+1}$, the position vectors of the vertices of $\Delta_{n+1}$.
\end{definition}

Some observations about this definition:
\begin{description}
  \item[First] \defref{d:simplex} provides a bijective map from probability distributions over $n+1$ outcomes, $\{p_i\}_{i=1}^{n+1}$, to $\Delta_{n+1}$.
  \item[Second] $\vec t_i^2=a$, all vertices are at equal distance from the origin, a necessary condition to obtain a regular simplex centered at the origin, and this distance is not zero so that the simplex is not a point.
  \item[Third] $\vec t_i\cdot\vec t_j=b$ when $i\neq j$, this implies that the vectors of $V$ spread uniformly in space. The condition $b\neq a$ is necessary so that the vectors in $V$ are different. Also we can deduce that $b\neq 0$ because we cannot have $n+1$ orthogonal and no null vectors in $\R^n$.
\end{description}

Now we will deduce some important properties.

\begin{proposition}\label{p:nbasis}
Every subset of $V$ that contains $n$ vectors is a basis of $\R^n$.
\end{proposition}
\begin{proof}
Given the symmetry between the vectors of $V$  we can, without loss of generality, prove the proposition for the subset $\{\vec t_i\}_{i=1}^{n}$. Therefore we have to prove that
\begin{equation}\label{e:indlint}
    \sum_{i=1}^{n}\lambda_i\vec t_i=\vec 0,
\end{equation}
has only the trivial solution. We multiply \eqref{e:indlint} by $\vec t_{n+1}$ and the result is
\begin{equation*}
    b\sum_{i=1}^{n}\lambda_i=0,
\end{equation*}
and because $b\neq 0$
\begin{equation}\label{e:sumlambda}
    \sum_{i=1}^{n}\lambda_i=0.
\end{equation}

Now we multiply \eqref{e:indlint} by $\vec t_j$, where $j\in\{1,2,\ldots,n\}$
\begin{align*}
    \sum_{i=1}^{n}\lambda_i\vec t_i\cdot\vec t_j&=0\\
    \sum_{i=1}^{n}\lambda_i[(a-b)\delta_{ij}+b]&=0\\
    (a-b)\lambda_j+b\sum_{i=1}^{n}\lambda_i&=0\\
    (a-b)\lambda_j&=0 &&\text{(by equation \eqref{e:sumlambda})}\\
    \lambda_j &=0, &&\text{(because $a\neq b$)}
\end{align*}
\end{proof}

\begin{proposition}
The sum of all the elements of $V$ is null.
\end{proposition}
\begin{proof}
By \propref{p:nbasis} we know that $S=\{\vec t_i\}_{i=1}^{n}$ is a basis of $\R^n$, so we can express $\vec t_{n+1}$ as a linear combination of S-vectors
\begin{equation}\label{e:coordtn+1}
    \vec t_{n+1}=\sum_{i=1}^{n}\lambda_i\vec t_i
\end{equation}
Multiplying \eqref{e:coordtn+1} by $\vec t_{n+1}$ we obtain
\begin{equation}\label{e:absumlambda}
    a=b\sum_{i=1}^{n}\lambda_i.
\end{equation}
Now we multiply \eqref{e:coordtn+1} by $\vec t_j$, with $j\in\{1,2,\ldots,n\}$
\begin{align*}
    b&=\sum_{i=1}^{n}\lambda_i[(a-b)\delta_{ij}+b]\\
    b&=(a-b)\lambda_j+a &&\text{(by equation \eqref{e:absumlambda})},
\end{align*}
and now is immediate that
\begin{equation*}
    \lambda_j=-1,\quad j\in\{1,2,\ldots,n\}.
\end{equation*}
Substituting in \eqref{e:coordtn+1}
\begin{equation*}
    \vec t_{n+1}=-\sum_{i=1}^{n}\vec t_i
\end{equation*}
or
\begin{equation*}
    \sum_{i=1}^{n+1}\vec t_i=\vec 0.
\end{equation*}
\end{proof}

\begin{corollary}
The relation between $a$ and $b$ is
\begin{equation}\label{e:b(a)}
    b=-\frac{a}{n}.
\end{equation}
\begin{proof}
If in equation \eqref{e:absumlambda} we substitute the values found for $\lambda_i$ we immediately obtain \eqref{e:b(a)}.
\end{proof}
\end{corollary}

Then $a$ remains as a free parameter that fixes the scale of $\Delta_{n+1}$. In view of the last corollary I think that a convenient value for $a$ is $a=n\Rightarrow b=-1$ because then many $1/n$ factors dissapear. With these values the inner product between vectors of $V$ reads as follows
\begin{equation}\label{e:titj}
    \vec t_i\cdot\vec t_j=(n+1)\delta_{ij}-1\quad i,j\in\{1,2,\ldots,n+1\}\,,
\end{equation}
that is, $\norm{\vec t_i}=\sqrt n$ and $\vec t_i\cdot\vec t_j=-1$ when $i\neq j$.

In a regular simplex all the m-facets, m-dimensional facets, are at the same distance from the origin.

\begin{proposition}
The distance from the origin to the m-facets is
\begin{equation}\label{e:dm}
    d_m=\sqrt{\frac{n-m}{m+1}}
\end{equation}
\end{proposition}
\begin{proof}
Because all m-facets are equidistant from the origin we can take one particular m-facet, for example that defined by the vectors $\{\vec t_i\}_{i=1}^{m+1}$. The centroid of this m facet is in the position
\begin{equation*}
    \vec F_m=\frac{1}{m+1}\sum_{i=1}^{m+1}\vec t_i
\end{equation*}
then
\begin{align*}
    d_m^2&=\vec F_m^2=\frac{1}{(m+1)^2}\sum_{i,j=1}^{m+1}\vec t_i\cdot\vec t_j\\
    &=\frac{1}{(m+1)^2}\sum_{i,j=1}^{m+1}\left[(n+1)\delta_{ij}-1\right] &&\text{(by equation \eqref{e:titj})}\\
    &=\frac{n-m}{m+1}
\end{align*}
\end{proof}

Among these distances are of special interest the radius of the inner sphere and the radius of the outer sphere. The outer sphere is the $(n-1)$-sphere centered at the origin that contains $\Delta_{n+1}$ and such that its radius is minimal, evidently its radius is the distance from the origin to the 0-facets (the vertices of the simplex)
\begin{equation}\label{e:Rout}
    R_{out}=d_0=\norm{\vec t_i}=\sqrt{n}.
\end{equation}
The inner sphere is the $(n-1)$-sphere contained in $\Delta_{n+1}$ centered at the origin and such that its radius is maximal. From \eqref{e:dm} we see that
\begin{equation}\label{e:d_mseq}
    0=d_n<d_{n-1}<\cdots<d_0=\sqrt{n}\,,
\end{equation}
then the inner sphere has radius
\begin{equation}\label{e:Rin}
    R_{in}=d_{n-1}=\frac{1}{\sqrt n}
\end{equation}
with a greater radius the sphere would have points situated beyond the $(n-1)$-facets, and therefore outside $\Delta_{n+1}$.

As we noted \defref{d:simplex} provides a bijective map from probability distributions over $n+1$ outcomes, $\{p_i\}_{i=1}^{n+1}$, to $\Delta_{n+1}$, the map is
\begin{align*}
    f_{\Delta_{n+1}}:\mathcal{P}_{n+1} &\rightarrow \Delta_{n+1}\\
    \{p_i\}_{i=1}^{n+1} &\mapsto \vec s=\sum_{i=1}^{n+1}p_i\vec t_i\,,
\end{align*}
where we denote this map by $f_{\Delta_{n+1}}$ because the vectors $\vec s$ depends on the election of simplex $\Delta_{n+1}$, we have also introduced the symbol $\mathcal{P}_{n+1}$ to denote the set of all probability distributions over $n+1$ outcomes.

In the next proposition we will see how to recover a probability distribution from a given vector $\vec s\in\Delta_{n+1}$.

\begin{proposition}
    The inverse of the map defined above is
    \begin{align}\label{e:finv}
        f_{\Delta_{n+1}}^{-1}: \Delta_{n+1} &\rightarrow \mathcal{P}_{n+1}\notag\\
        \vec s &\mapsto \left\{p_i =\frac{1}{n+1}(\vec s\cdot \vec t_i+1)\right\}_{i=1}^{n+1}
    \end{align}
\end{proposition}
\begin{proof}
    \begin{align*}
        \vec s &= \sum_{j=1}^{n+1}p_j\vec t_j\\
        \vec s\cdot \vec t_i &= \sum_{j=1}^{n+1}p_j\vec t_j\cdot \vec t_i\\
        \vec s\cdot \vec t_i &= \sum_{j=1}^{n+1}p_j[(n+1)\delta_{ij}-1]&\text{(by \eqref{e:titj})}\,,
    \end{align*}
Then
\begin{equation}\label{e:p(v)}
    p_i =\frac{1}{n+1}(\vec s\cdot \vec t_i+1)
\end{equation}
\end{proof}

It is also interesting the following relation

\begin{proposition}
    If $\vec s\in\Delta_{n+1}$ and $\{p_i\}_{i=1}^{n+1}$ is its corresponding probability distribution, then
    \begin{equation}\label{e:p2(s)}
        \sum_{i=1}^{n+1}p_i^2=\frac{1}{n+1}(\vec s^2+1)
    \end{equation}
\end{proposition}
\begin{proof}
    \begin{align*}
        \vec s^2 &= \sum_{i,j=1}^{n+1}p_ip_j\vec t_i\cdot\vec t_j\\
        &= \sum_{i,j=1}^{n+1}p_ip_j [(n+1)\delta_{ij}-1] &&\text{(by \eqref{e:titj})}\\
        &= (n+1)\sum_{i=1}^{n+1}p^2_i-1
    \end{align*}
From which it follows the proposition.
\end{proof}

\begin{corollary}\label{c:p2facet}
    When $\vec s^2=d_m^2$ then
    \begin{equation}\label{e:p2mfacet}
        \sum_{i=1}^{n+1}p_i^2=\frac{1}{m+1}
    \end{equation}
\end{corollary}

\section{Bloch representation of quantum states}\label{s:BlochQS}
We will denote by $\mathcal{D}_d$ the set of density matrices of order $d$, namely
\begin{equation}\label{e:defD}
    \mathcal{D}_d=\{\rho\in\mathcal{M}_d(\mathbb{C})\; : \;\rho=\rho^{\dag}\, ,\: \rho\geq 0\;\text{and}\;\Tr\rho=1\}.
\end{equation}
Any d-dimensional density matrix $\rho$ can be represented as \cite{FQIT}
\begin{equation}\label{e:bloch}
    \rho=\frac{1}{d}+\sqrt{\frac{d+1}{2d}}\;\vec r \cdot \boldsymbol\sigma\,,
\end{equation}
where the coefficient $\sqrt{\frac{d+1}{2d}}$ has been chosen for later convenience. If the vector, $\vec r\in \R^{d^2-1}$, in \eqref{e:bloch} is such that the corresponding $\rho$ is a true density matrix, then this vector will be called the Bloch vector associated to $\rho$. The set of all Bloch vectors in $\R^{d^2-1}$ will be denoted by $\mathcal{B}_{d^2-1}$. The components of the vector $\boldsymbol\sigma=(\sigma_1,\sigma_2,\ldots,\sigma_{d^2-1})$ are hermitian matrices with null trace which form a basis of the algebra $\mathfrak{su}(d)$ and we follow the convention that $scalar+square\: matrix$ is read as $scalar I+square\: matrix$, where $I$ is the identity matrix of the appropriate order.

Some useful formulae involving the $\sigma_a$'s are, see for example appendix 2 of \cite{GQS},
\begin{align}
    [\sigma_a,\sigma_b] &=2\im f_{abc}\sigma_c\,,\label{e:sigmacommutator}\\
    \{\sigma_a,\sigma_b\} &=\frac{4}{d}\:\delta_{ab}+2\:d_{abc}\sigma_c\,,\label{e:sigmaanticommutator}\\
    \Tr(\sigma_a\sigma_b) &=2\:\delta_{ab}\,,\label{e:trsigmasigma}\\
    \Tr(\sigma_a\sigma_b\sigma_c) &=2\:d_{abc}+2\im f_{abc}\label{e:trsigmasigmasigma},
\end{align}
where we sum over repeated indices, $f_{abc}\in\R$ is totally antisymmetric, $d_{abc}\in\R$ is totally symmetric, traceless and is identically null when $d=2$.

In the next proposition we see what conditions have to be fulfilled by $\vec r$ so that $\rho$ is a pure state. This proposition is enunciated in \cite{GQS} (eq. 8.24, p. 215) although with a different normalization for the Bloch vectors.

\begin{proposition}
    $\rho$ is a pure state if and only if the associated Bloch vector, $\vec r$, satisfies
        \begin{equation}\label{e:normr}
            \vec r^2=\frac{d-1}{d+1}
        \end{equation}
    and
        \begin{equation}\label{e:r*r}
            \vec r\ast\vec r=(d-2)\sqrt{\frac{2}{d(d+1)}}\:\vec r\,,
        \end{equation}
    where $(\vec r\ast\vec r)_c=d_{abc}\,r_a\,r_b$
\end{proposition}
\begin{proof}
$\rho$ is a pure state if and only if $\rho^2=\rho$
\begin{align*}
    \rho^2 &= \frac{1}{d^2}+\frac{d+1}{2d}\;r_a r_b \sigma_a\sigma_b+\frac{1}{d}\sqrt{\frac{2(d+1)}{d}}\vec r \cdot \boldsymbol\sigma\\
    &=\frac{1}{d^2}+\frac{d+1}{4d}\;r_a r_b \{\sigma_a,\sigma_b\}+\frac{1}{d}\sqrt{\frac{2(d+1)}{d}}\vec r \cdot \boldsymbol\sigma &&\text{($r_ar_b$ is symmetric in $a,b$)}\\
    &=\frac{1}{d^2}+\frac{d+1}{4d}\;r_a r_b\left[\frac{4}{d}\:\delta_{ab}+2\:d_{abc}\sigma_c\right] +\frac{1}{d}\sqrt{\frac{2(d+1)}{d}}\vec r \cdot \boldsymbol\sigma &&\text{(by equation \eqref{e:sigmaanticommutator})}\\
    &=\frac{1}{d^2}+\frac{d+1}{d^2}\;\vec r^2+\frac{d+1}{2d}(\vec r\ast\vec r)\cdot\boldsymbol\sigma +\frac{1}{d}\sqrt{\frac{2(d+1)}{d}}\vec r \cdot \boldsymbol\sigma\\
    &=\frac{1}{d}\left(\frac{1}{d}+\frac{d+1}{d}\;\vec r^2\right)+\sqrt{\frac{d+1}{2d}}\left(\sqrt{\frac{d+1}{2d}}(\vec r\ast\vec r)+\frac{2}{d}\vec r\right) \cdot \boldsymbol\sigma
\end{align*}
Now we impose that this linear combination of $I_d$, the identity matrix of order $d$, and the $\sigma_a$ matrices is equal to the linear combination of these same matrices in \eqref{e:bloch}. Because the $\sigma_a$ matrices are traceless and they satisfy \eqref{e:trsigmasigma}, we see that the matrices of the set $\{I_d,\sigma_1,\sigma_2,\ldots,\sigma_{d^2-1}\}$ are orthogonal, with respect to the Hilbert-Schmidt inner product, so they are linearly independent, in fact, they form a basis of the hermitian matrices of order $d$, then two linear combinations of these matrices are equal if and only if its coefficients are equal. In our case this means that
\begin{align*}
    \frac{1}{d}+\frac{d+1}{d}\;\vec r^2 &=1 &&\text{and}\\
    \sqrt{\frac{d+1}{2d}}(\vec r\ast\vec r)+\frac{2}{d}\vec r &=\vec r
\end{align*}
From which we obtain \eqref{e:normr} and \eqref{e:r*r}.
\end{proof}

Equation \eqref{e:bloch} defines a bijective map
\begin{align}\label{e:qmap}
    q_{\boldsymbol\sigma}:\mathcal{B}_{d^2-1} &\rightarrow \mathcal{D}_d\notag\\
    \vec r &\mapsto \rho=\frac{1}{d}+\sqrt{\frac{d+1}{2d}}\;\vec r \cdot \boldsymbol\sigma\,,
\end{align}
We write $q_{\boldsymbol\sigma}$ because this map is fixed once we have chosen the basis for $\mathfrak{su}(d)$. It is interesting to find its inverse

\begin{proposition}
    The inverse of the map $q_{\boldsymbol\sigma}$ is
    \begin{align}\label{e:qmapinv}
    q_{\boldsymbol\sigma}^{-1}:\mathcal{D}_d &\rightarrow \mathcal{B}_{d^2-1}\notag\\
    \rho &\mapsto \vec r=\sqrt{\frac{d}{2(d+1)}}\Tr(\rho\boldsymbol\sigma).
\end{align}
\begin{proof}
    \begin{align*}
    \rho &= \frac{1}{d}+\sqrt{\frac{d+1}{2d}}\;r_a \sigma_a \\
    \Tr(\rho\sigma_b) &=\Tr\left(\frac{1}{d}\sigma_b+\sqrt{\frac{d+1}{2d}}\;r_a \sigma_a\sigma_b\right)\\
    \Tr(\rho\sigma_b) &=\sqrt{\frac{d+1}{2d}}\;r_a\Tr(\sigma_a\sigma_b)&&\text{(because $\sigma_b$ is traceless)}\\
    \Tr(\rho\sigma_b) &=\sqrt{\frac{2(d+1)}{d}}\;r_b  &&\text{(by \eqref{e:trsigmasigma})}
    \end{align*}
\end{proof}
\end{proposition}

\section{Bloch representation of a SIC-POVM}\label{s:SICBloch}
First we define a SIC-POVM, see for example \cite{SICQM}
\begin{definition}
    A set of positive operators $\{E_i\}_{i=1}^{d^2}$ is a SIC-POVM, for $d$-dimensional systems, if the following conditions are satisfied
\begin{align}\label{e:defSIC}
    E_i &=\frac{1}{d}\rho_i\,,\quad\text{with $\rho_i$ a pure state and $i\in\{1,2,\ldots,d^2\}$}.\notag\\
    \sum_{i=1}^{d^2}E_i &= 1\,.\\
    \Tr(E_jE_j)&=\frac{d\delta_{ij}+1}{d^2(d+1)}\,,\quad\text{with $i,j\in\{1,2,\ldots,d^2\}$}\notag.
\end{align}
\end{definition}

From this definition and from the last section we see that the elements of the SIC-POVM can be represented in the following way
\begin{equation}\label{e:E_iBloch}
    E_i=\frac{1}{d^2}+\frac{1}{d}\sqrt{\frac{d+1}{2d}}\;\vec e_i \cdot \boldsymbol\sigma\,.
\end{equation}
Where each $\vec e_i$ satisfies \eqref{e:normr} and \eqref{e:r*r}. Our next task is to find what other conditions vectors $\vec e_i$ satisfy so that the corresponding operators $E_i$ form a SIC-POVM.
\begin{proposition}
    The vectors $\{\vec e_i\}_{i=1}^{d^2}$ defined in \eqref{e:E_iBloch} are the positions of the vertices of a regular and $(d^2-1)$-dimensional simplex.
\end{proposition}
\begin{proof}
    \begin{align*}
    \Tr(E_iE_j)&=\Tr\left[\left(\frac{1}{d^2}+\frac{1}{d}\sqrt{\frac{d+1}{2d}}\;\vec e_i \cdot \boldsymbol\sigma\right)\left(\frac{1}{d^2}+\frac{1}{d}\sqrt{\frac{d+1}{2d}}\;\vec e_j \cdot \boldsymbol\sigma\right)\right]\\
    &=\Tr\left(\frac{1}{d^4}+\frac{d+1}{2d^3}\,e_{ia}\,e_{jb}\,\sigma_a\,\sigma_b\right)\\
    &=\frac{1}{d^3}+\frac{d+1}{d^3}\,\vec e_i\cdot\vec e_j \quad\text{(by \eqref{e:trsigmasigma})}\\
    &=\frac{d\delta_{ij}+1}{d^2(d+1)} \qquad\text{(Imposing the last condition of \eqref{e:defSIC})}.
\end{align*}
which implies
\begin{equation*}
    \vec e_i\cdot\vec e_j=\frac{d^2\delta_{ij}-1}{(d+1)^2},
\end{equation*}
and we see that the vectors $\{\vec e_i\}_{i=1}^{d^2}$ satisfy \defref{d:simplex} with $n=d^2-1$, $a=(d-1)/(d+1)$, as it should be by \eqref{e:normr}, and $b=-1/(d+1)^2$.
\end{proof}

\section{Probability distributions corresponding to quantum states}\label{s:PQS}
Let
\begin{equation}\label{e:rho(s)}
    \rho=\frac{1}{d}+\sqrt{\frac{d+1}{2d}}\;\vec r \cdot \boldsymbol\sigma
\end{equation}
be a quantum state, mixed or pure. We will find what is the probability distribution over the outcomes of a SIC-POVM, and what is the vector in the corresponding simplex of probability, as described in \secref{s:simplexes}. The probability distribution is
\begin{align}\label{e:p(er)}
    p_i&=\Tr(E_i\rho)\notag\\
    &=\Tr\left [\left(\frac{1}{d^2}+\frac{1}{d}\sqrt{\frac{d+1}{2d}}\;\vec e_i \cdot \boldsymbol\sigma\right)\left(\frac{1}{d}+\sqrt{\frac{d+1}{2d}}\;\vec r \cdot \boldsymbol\sigma\right)\right]\notag\\
    &=\Tr\left(\frac{1}{d^3}+\frac{d+1}{2d^2}\,e_{ia}\,r_b\,\sigma_a\,\sigma_b\right)\notag\\
    &=\frac{1}{d^2}+\frac{d+1}{d^2}\vec e_i\cdot \vec r
\end{align}
Naming by $\mathcal{Q}_{\mathcal{P}}$ the set of all probability distributions over the outcomes of the SIC-POVM $\{E_i\}_{i=1}^{d^2}$ that correspond to a quantum state, then \eqref{e:p(er)} defines the following bijective map
\begin{align}\label{e:mE}
    m_E: \mathcal{B}_{d^2-1} &\rightarrow \mathcal{Q}_{\mathcal{P}}\subset \mathcal{P}_{d^2}\notag\\
    \vec r &\mapsto \left\{p_i=\frac{1}{d^2}+\frac{d+1}{d^2}\vec e_i\cdot \vec r\right\}_{i=1}^{d^2}
\end{align}

To represent this distribution we can choose any regular simplex, $\Delta_{d^2}$, in $\R^{d^2-1}$, the most natural is the one defined by the vectors
\begin{equation*}
    \vec t_i=(d+1)\vec e_i\qquad i\in\{1,2,\ldots,d^2\},
\end{equation*}
because, as we have seen, the vectors $\vec e_i$ define themselves a simplex, the factor is needed so that $\norm{\vec t_i}=\sqrt{d^2-1}$ in accordance with the norm used in \secref{s:simplexes} for these vectors. Now we make use of the bijective map $f_{\Delta_{d^2}}$ defined in \secref{s:simplexes} to define the set $\mathcal{Q}=f_{\Delta_{d^2}}(\mathcal{Q}_{\mathcal{P}})$, this set contains therefore all elements of $\Delta_{d^2}$ corresponding to quantum states. We then have the bijective map $f_{\Delta_{d^2}}|_{\mathcal{Q}_{\mathcal{P}}}$ that we will denote by $g_{\mathcal{Q}_{\mathcal{P}}}$ and is defined therefore as
\begin{align}\label{e:frest}
    g_{\mathcal{Q}_{\mathcal{P}}}:\mathcal{Q}_{\mathcal{P}} &\rightarrow \mathcal{Q}\notag\\
    \{p_i\}_{i=1}^{d^2} &\mapsto \vec s=\sum_{i=1}^{d^2}p_i\vec t_i\:.
\end{align}

Now we can prove the main result of this paper.

\begin{proposition}
    The set of Bloch vectors, and the set of elements of $\Delta_{d^2}$ corresponding to quantum states are the same set. Namely
    \begin{equation}\label{e:B=Q}
        \mathcal{B}_{d^2-1}=\mathcal{Q}\,.
    \end{equation}
\end{proposition}
\begin{proof}
    The map $(g_{\mathcal{Q}_{\mathcal{P}}}\circ m_E)$ is a bijection because it is a composition of bijections, it goes from $\mathcal{B}_{d^2-1}\subset\R^{d^2-1}$ to $\mathcal{Q}\subset\R^{d^2-1}$. We therefore need to prove that this map is the identity map, that is
    \begin{equation*}
        (g_{\mathcal{Q}_{\mathcal{P}}}\circ m_E)(\vec r)=\vec r\qquad\forall\vec r\in\mathcal{B}_{d^2-1}\,,
    \end{equation*}
    or, equivalently
    \begin{equation*}
        m_E(\vec r)=g_{\mathcal{Q}_{\mathcal{P}}}^{-1}(\vec r)\,.
    \end{equation*}
    Applying \eqref{e:mE} and \eqref{e:finv} we obtain
    \begin{align*}
        \left\{\frac{1}{d^2}+\frac{d+1}{d^2}\vec e_i\cdot \vec r\right\}_{i=1}^{d^2} &= \left\{\frac{1}{d^2}(\vec t_i\cdot \vec r+1)\right\}_{i=1}^{d^2}\\
        \left\{\frac{1}{d^2}[(d+1)\vec e_i\cdot \vec r+1]\right\}_{i=1}^{d^2} &= \left\{\frac{1}{d^2}(\vec t_i\cdot \vec r+1)\right\}_{i=1}^{d^2}\,,
    \end{align*}
    and this last equality is true for all $\vec r\in\mathcal{B}_{d^2-1}$ because $\vec t_i=(d+1)\vec e_i$.
\end{proof}

Now we can study basic facts about the shape of $\mathcal{Q}$. We will denote by $\mathfrak{P}$ the subset of $\mathcal{Q}$ corresponding to pure states.

\begin{corollary}
    $\mathfrak{P}$ is a subset of the $(d^2-2)$-sphere of radius $R_{\mathfrak{P}}=\sqrt{\frac{d-1}{d+1}}$.
\end{corollary}
\begin{proof}
    Because \eqref{e:B=Q} we can use \eqref{e:normr} that gives the norm of Bloch vectors corresponding to pure states.
\end{proof}

\begin{corollary}
The sphere that contains $\mathfrak{P}$ is not completely inside $\Delta_{d^2}$, except in the case $d=2$, then it is the inner sphere of the simplex.
\end{corollary}
\begin{proof}
We have to prove that the next inequality is true, and that is an equality only if $d=2$.
    \begin{align*}
        R^2_{\mathfrak{P}} &\geq R^2_{in}\\
        \frac{d-1}{d+1} &\geq \frac{1}{d^2-1} &&\text{(by \eqref{e:Rin})}\\
        d-1 &\geq \frac{1}{d-1} &&\text{(multiplying by $d+1$)}\\
        (d-1)^2 &\geq 1 &&\text{(multiplying by $d-1$)}
    \end{align*}
\end{proof}

The following result was also obtained in \cite{AEF2009}.

\begin{proposition}
    The sphere that contains $\mathfrak{P}$ is tangent to the facets of $\Delta_{d^2}$ of dimension $m_{\mathfrak{P}}=\frac{(d+2)(d-1)}{2}$.
\end{proposition}
\begin{proof}
    First observe that $m_{\mathfrak{P}}$ is a natural number because $d+2$ and $d-1$ have opposite parity, so one of them is divisible by 2. Equation \eqref{e:dm} gives de distance from the origin to the m-facets, we will demonstrate that $d_m$, for the particular value $m=m_{\mathfrak{P}}$,  is the radius of the sphere that contains $\mathfrak{P}$.
    \begin{align*}
        d_{m_{\mathfrak{P}}} &= \sqrt{\frac{d^2-1-m_{\mathfrak{P}}}{m_{\mathfrak{P}}+1}}\\
        &= \sqrt{\frac{d^2-1-\frac{(d+2)(d-1)}{2}}{\frac{(d+2)(d-1)}{2}+1}}\\
        &= \sqrt{\frac{d-1}{d+1}}
    \end{align*}
\end{proof}

Therefore $\mathcal{Q}$ \emph{is a subset} of a $(d^2-1)$-ball truncated by the $m$-facets of $\Delta_{d^2-1}$ with $m>m_{\mathfrak{P}}=\frac{(d+2)(d-1)}{2}$, because from \eqref{e:d_mseq} we have that if $m>m_{\mathfrak{P}}$ then $d_m<d_{m_{\mathfrak{P}}}$. But the shape of $\mathcal{Q}$ is not simply this truncated ball, as we have emphasized it is a proper subset of this body. Remember that the pure states are a $(2d-2)$-dimensional manifold, so not all points on the surface of the ball, even those situated inside the simplex, can be quantum states.

The following result can be found in \cite{QBRQS}, although it is obtained in a different way.

\begin{corollary}
    If $\{p_i\}_{i=1}^{d^2}$ is the distribution of probability over the outcomes of a SIC-POVM of a pure state then
    \begin{equation}\label{e:p2pure}
        \sum_{i=1}^{d^2}p_i^2=\frac{2}{d(d+1)}\,.
    \end{equation}
\end{corollary}
\begin{proof}
    We simply use \corref{c:p2facet} and the last proposition.
    \begin{align*}
        \sum_{i=1}^{d^2}p_i^2 &= \frac{1}{m_{\mathfrak{P}}+1}\\
        &= \frac{1}{\frac{(d+2)(d-1)}{2}+1}\\
        &= \frac{2}{d(d+1)}\,.
    \end{align*}
\end{proof}

\section{Conclusions and future research}
We have proved that with suitable scalings the set of Bloch vectors, $\mathcal{B}$, is equal to the set of points, $\mathcal{Q}$, of the simplex associated to the probability distributions over the outcomes of a SIC-POVM that correspond to quantum states. We have see that $\mathcal{Q}$ is a subset of a $d^2-1$-ball truncated by the $m$-facets of a $d^2-1$-simplex with $m>m_{\mathfrak{P}}=\frac{(d+2)(d-1)}{2}$ ($d$ is the dimension of the Hilbert space we are considering). As a consequence for pure states $\sum_{i=1}^{d^2}p_i^2=\frac{2}{d(d+1)}$, where $p_i$ is the probability of obtaining result $i$ when measuring the SIC-POVM $\{E_i\}_{i=1}^{d^2}$.

The final objective of this work is to understand in physical, not purely mathematical, terms why $\mathcal{Q}$ has that structure. Why are the pure states situated on a sphere? Why is this sphere tangent to some of the facets of our simplex $\Delta_{d^2}$? Why are these facets precisely those of dimension $\frac{(d+2)(d-1)}{2}$?.

I think that trying to answer this questions we will have a deeper understanding of Quantum Foundations, and therefore of our world.

\subsection*{Acknowledgment}
I want to acknowledge Dr. Chris Fuchs, that kindly give me his endorsement to submit papers to the arXiv.

\bibliographystyle{siam}
\bibliography{MyBib}

\end{document}